\documentclass[a4paper, 12pt]{article}

\usepackage[sort&compress]{natbib}
\bibpunct{(}{)}{;}{a}{}{,} 

\usepackage{amsthm, amsmath, amssymb, mathrsfs, multirow, url, subfigure}
\usepackage{graphicx} 
\usepackage{ifthen} 
\usepackage{amsfonts}
\usepackage[usenames]{color}
\usepackage{fullpage}

\theoremstyle{plain}

\newtheorem{prop}{Proposition}

\theoremstyle{definition}

\theoremstyle{remark}

\newcommand{\PP}{\mathbb{P}}

\newcommand{\X}{\mathcal{X}}

\newcommand{\E}{\mathsf{E}}

\renewcommand{\phi}{\varphi}

\newcommand{\unif}{\mathsf{Unif}}

\newcommand{\Fhat}{\widehat{F}}

\title{Fast nonparametric near-maximum likelihood estimation of a mixing density}
\author{Minwoo Chae,\footnote{Department of Mathematics, Applied Mathematics, and Statistics, Case Western University, {\tt minwooo.chae@gmail.com}} \quad Ryan Martin,\footnote{Department of Statistics, North Carolina State University, {\tt rgmarti3@ncsu.edu}} \quad and \quad Stephen G.~Walker\footnote{Department of Mathematics, University of Texas at Austin, {\tt s.g.walker@math.utexas.edu}}}
\date{\today}

\begin{document}

\maketitle 


\begin{abstract}
Mixture models are regularly used in density estimation applications, but the problem of estimating the mixing distribution remains a challenge.  Nonparametric maximum likelihood produce estimates of the mixing distribution that are discrete, and these may be hard to interpret when the true mixing distribution is believed to have a smooth density.  In this paper, we investigate an algorithm that produces a sequence of smooth estimates that has been conjectured to converge to the nonparametric maximum likelihood estimator.  Here we give a rigorous proof of this conjecture, and propose a new data-driven stopping rule that produces smooth near-maximum likelihood estimates of the mixing density, and simulations demonstrate the quality empirical performance of this estimator.  

\smallskip


\end{abstract}

\section{Introduction}
\label{S:intro}

Consider a mixture model with density function $f=f_P$ given by 
\begin{equation}
\label{eq:mixture}
f(y) = \int k(y \mid x) \, P(dx)
\end{equation}
where $k(y \mid x)$ is a known kernel density and $P$ is an unknown mixing distribution.  The goal is estimation of $P$ based on independent and identically distributed data $Y_1,\ldots,Y_n$ from the mixture $f$ in \eqref{eq:mixture}, a classically challenging problem in statistics.  If $P$ is a discrete distribution with fixed and finite number of components, then \eqref{eq:mixture} is a finite mixture model and is relatively straightforward; indeed, maximum likelihood computation is feasible with the EM algorithm \citep{dlr} and the usual asymptotic theory is available \citep{rednerwalker}.  The catch is that the number of mixture components can be difficult to specify.  Therefore, there has also been a lot of work on finite mixtures with an unknown number of components \citep[e.g.,][]{miller2016mixture, rousseau2011asymptotic, leroux, james, woosriram2006, richardson, roederwasserman}.  

Likelihood-based methods for estimating $P$ are available even without explicitly making the problem finite-dimensional.  Indeed, the likelihood function for $P$ in the nonparametric case is 
\begin{equation}
\label{eq:likelihood}
L(P) = \prod_{i=1}^n \int k(Y_i \mid x) \, P(dx), 
\end{equation}
and it is known that the maximizer $\hat P$, the nonparametric maximum likelihood estimator (MLE), is discrete, with at most $n$ components \citep{lindsay1983, lindsay1995}.  Discreteness simplifies computation, and fast algorithms are available, e.g., \citet{wang} and \citet{koenker.mizera.2014}.  However, if $P$ is believed to have a density with respect to, say, the Lebesgue measure, then the discrete estimator may not be satisfactory.  To remedy this, various smoothed versions of $\hat P$ have been proposed \citep[e.g.,][]{eggermont1995, eggermont1997, silverman1990, green1990, laird.louis.1991, liulevinezhu2009}, but some of these are rather complicated and there seems to be no general consensus that one smoothing method is any better than another.  

For Bayesian mixture models, the Dirichlet process prior \citep{ferguson1973} and variants of its stick-breaking representation \citep{sethuraman1994} have become a mainstay, largely because of the plethora of powerful Markov chain Monte Carlo methods available for evaluating the corresponding posterior \citep[e.g.,][]{escobar.west.1995, maceachern1998, dunsonpark2007, walker2007.slice, kalli.griffin.walker.2011}.  The focus of these developments, however, has been the mixture density, with the mixing distribution serving merely as a modeling tool; but see \citet{nguyen2013}.  As with the nonparametric MLE, the inherent discreteness of stick-breaking priors, while advantageous for mixture density estimation and modeling latent structures, is problematic in the context of nonparametric Bayesian estimation of a mixing density.    

A Bayesian-style recursive estimate for $P$, called \emph{predictive recursion}, was proposed by \citet{newton02} and studied theoretically by \citet{martinghosh}, \citet{tmg}, and \citet{mt-rate}.  The algorithm is fast and provides an estimator having a smooth density with respect to any specified dominating measure.  However, its dependence on the (arbitrary) order in which the data $Y_1,\ldots,Y_n$ is processed, hence it is not a Bayesian estimator, along with its inability to be characterized as an optimizer of any objective function, makes the predictive recursion estimator difficult to interpret.  


In this paper, we investigate properties of a simple and fast iterative algorithm for estimating $p$, one that shares certain features with the MLE, a Bayesian approach, as well as predictive recursion.  Versions of this algorithm have been presented in the literature before, and its convergence properties have been conjectured but not rigorously proved.  Here we fill this gap by providing a proof that algorithm converges to the nonparametric MLE as the number of iterations approaches infinity.  While the limit is a discrete distribution, it is interesting that at every finite number of iterations, the algorithm returns a continuous density.  This suggests that a smooth {\em near-MLE} of the density can be obtained by stopping the algorithm before convergence is achieved, and we propose an data-driven stopping rule and demonstrate empirically the quality performance of this nonparametric near-MLE of the mixing density compared to predictive recursion.

\section{A simple and fast algorithm}
\label{S:algorithm}

\subsection{Definition}
\label{SS:definition}

As discussed above, our focus is on estimating a smooth mixing density $p$ associated with the mixing distribution $P$ in \eqref{eq:mixture}; throughout, we assume that $p$ is a density with respect to Lebesgue measure, though other choices of dominating measure could be handled similarly.  Given a prior guess $p_0$ of $p$, if a data point $Y_i$ is observed, then the Bayesian update of $p_0$ to $p_{1,i}$, say, is 
\begin{equation}
\label{eq:single.obs}
p_{1,i}(x) = \frac{k(Y_i \mid x) p_0(x)}{f_0(Y_i)}, 
\end{equation}
where $f_0(y) = \int k(y \mid x) p_0(x) \,dx$.  However, we can carry out this single-observation update for any $i=1,\ldots,n$ and, since observations ordering is irrelevant, it is reasonable to take an average: 
\[ p_1(x) = \frac1n \sum_{i=1}^n p_{1,i}(x) = \frac1n \sum_{i=1}^n \frac{k(Y_i \mid x) p_0(x)}{f_0(Y_i)}. \]
This same argument can be applied, with $p_0$ replaced by $p_1$, to get an updated estimate $p_2$, and so on.  This suggests the following iterative algorithm for an estimator of $p$:
\begin{equation}
\label{eq:algorithm}
p_{t+1}(x) = \frac1n \sum_{i=1}^n \frac{k(Y_i \mid x) p_t(x)}{f_t(Y_i)}, \quad t \geq 0, 
\end{equation}
where $f_t(y) = \int k(y \mid x) p_t(x) \,dx$ for general $t \geq 0$.  Algorithms similar to \eqref{eq:algorithm} for certain applications or models have appeared in the literature; see, e.g., \citet{vardi.shepp.kaufman.1985}, \citet{laird.louis.1991}, and \citet{vardi.lee.1993}.  But despite the hints in these papers about more general versions, it seems that the algorithm \eqref{eq:algorithm} has not been studied thoroughly and in the level of generality considered here.  

Aside from this Bayesian-motivated formulation, there are number of ways to think about this algorithm and understand what it is trying to do.  First, note the similarities with the predictive recursion algorithm of \citet{newton02} which updates by taking a weighted average of the current guess and the single-observation Bayes update \eqref{eq:single.obs} based on the current guess as the prior.  These computations proceed along the sequence $i=1,\ldots,n$ and, therefore, the result depends on the arbitrary order of the data sequence.  The proposed algorithm \eqref{eq:algorithm} can, therefore, be viewed as an order-invariant version of predictive recursion that can also be refined {\em ad infinitum}, by taking $t \to \infty$, if desired.  

Second, suppose that $p_t$ converges, in some sense, to a limit $p_\infty$.  Then that limit, in particular, its corresponding mixture $f_\infty$, must satisfy 
\[ \frac1n \sum_{i=1}^n \frac{k(Y_i \mid x)}{f_\infty(Y_i)} = 1 \quad \forall \; x. \]
This condition boils down to one involving the directional derivative of the log-likelihood function $\ell(P) = \log L(P)$, for $L$ as in \eqref{eq:likelihood}, which, according to Theorem~19 in \citet{lindsay1995}, characterizes the nonparametric MLE.  Therefore, we can expect that, if the limit exists, then it must be the nonparametric MLE, and details supporting this claim are presented in Section~\ref{SS:properties}.  However, for all finite $t \geq 0$, the continuity of the initial guess $p_0$ persists, so stopping at some $T < \infty$ makes $p_T$ a sort of ``smoothed'' nonparametric MLE; see Section~\ref{S:smoothed}.

\subsection{Illustration}
\label{SS:poisson}

Example~1.2 in \citet{bohning} presents data on the number of illness spells for 602 pre-school children in Thailand over a two-week period; see, also, \citet[][Table~1]{wang}.  B\"ohning argues that a Poisson model provides poor fit, so he proposes a nonparametric Poisson mixture model instead.  Wang produces the nonparametric MLE for the mixing distribution displayed in the first five panels Figure~\ref{fig:thai} (vertical gray lines).  These panels also show the mixing density estimates $p_T$ from \eqref{eq:algorithm} for five different stopping times, namely, $T \in \{5, 10, 100, 500, 5000\}$, based on $p_0 = \unif(0,25)$.  It is interesting that the estimate very quickly forgets the shape of the uniform initial guess, that at 500 iterations it has mostly identified the locations of the nonparametric MLE, but that it takes many more iterations before it starts to clearly concentrate around those locations.  Most interesting, however, is that the log-likelihood trajectory $t \mapsto \ell(p_t)$ is non-decreasing and very quickly jumps up to near $\max_P \ell(P)$, the maximum possible value; in fact, the relative difference between $\ell_n(p_{10})$ and the maximum is about 0.003, close to satisfying any reasonable convergence criterion.  So, at least in the sense of the likelihood value, $p_{10}$ is virtually as good as $\hat p$, and can be computed very quickly and simply; plus, $p_{10}$ is also a continuous density.  

\begin{figure}
\begin{center}
\subfigure[$T=5$ iterations]{\scalebox{0.6}{\includegraphics{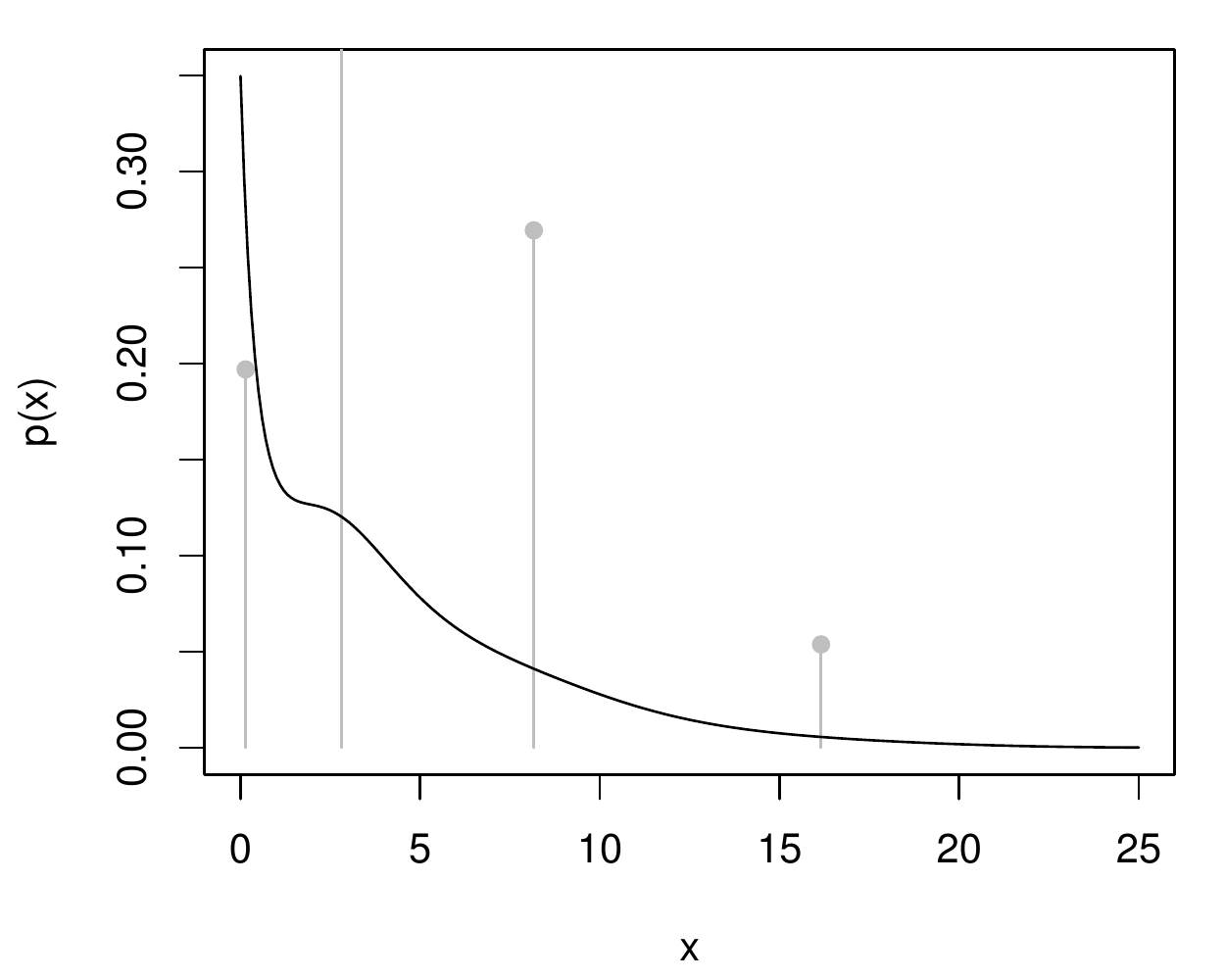}}}
\subfigure[$T=10$ iterations]{\scalebox{0.6}{\includegraphics{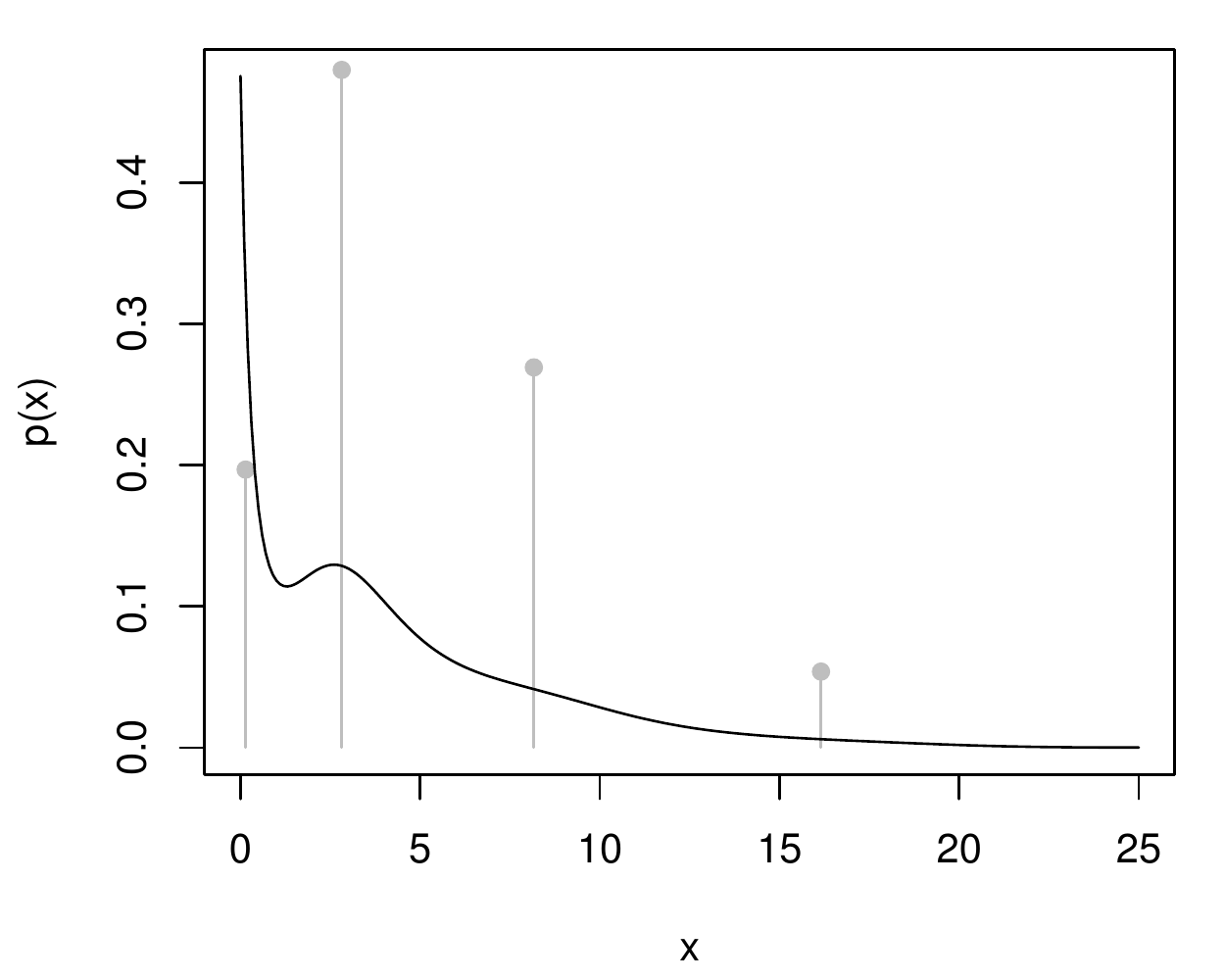}}}
\subfigure[$T=100$ iterations]{\scalebox{0.6}{\includegraphics{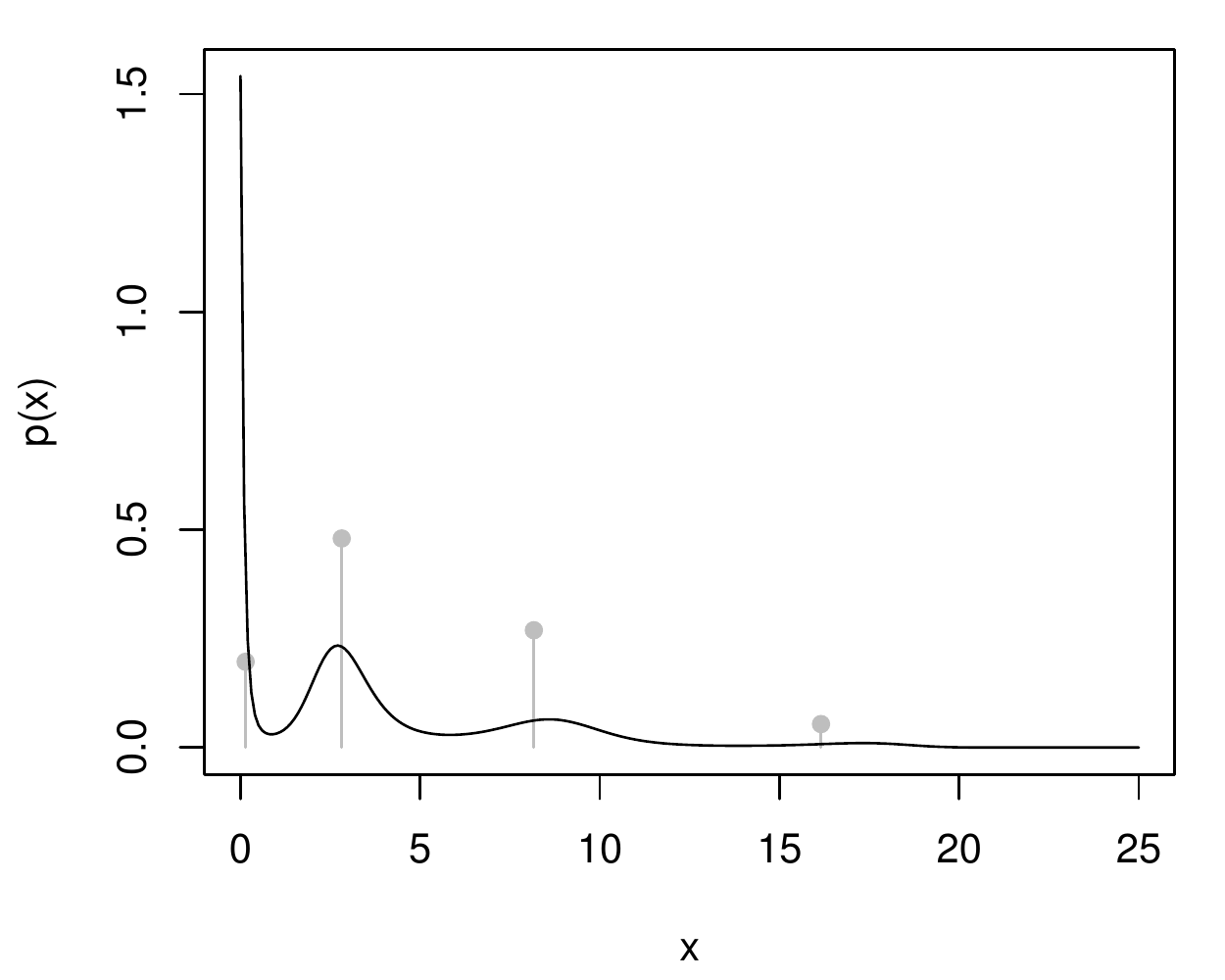}}}
\subfigure[$T=500$ iterations]{\scalebox{0.6}{\includegraphics{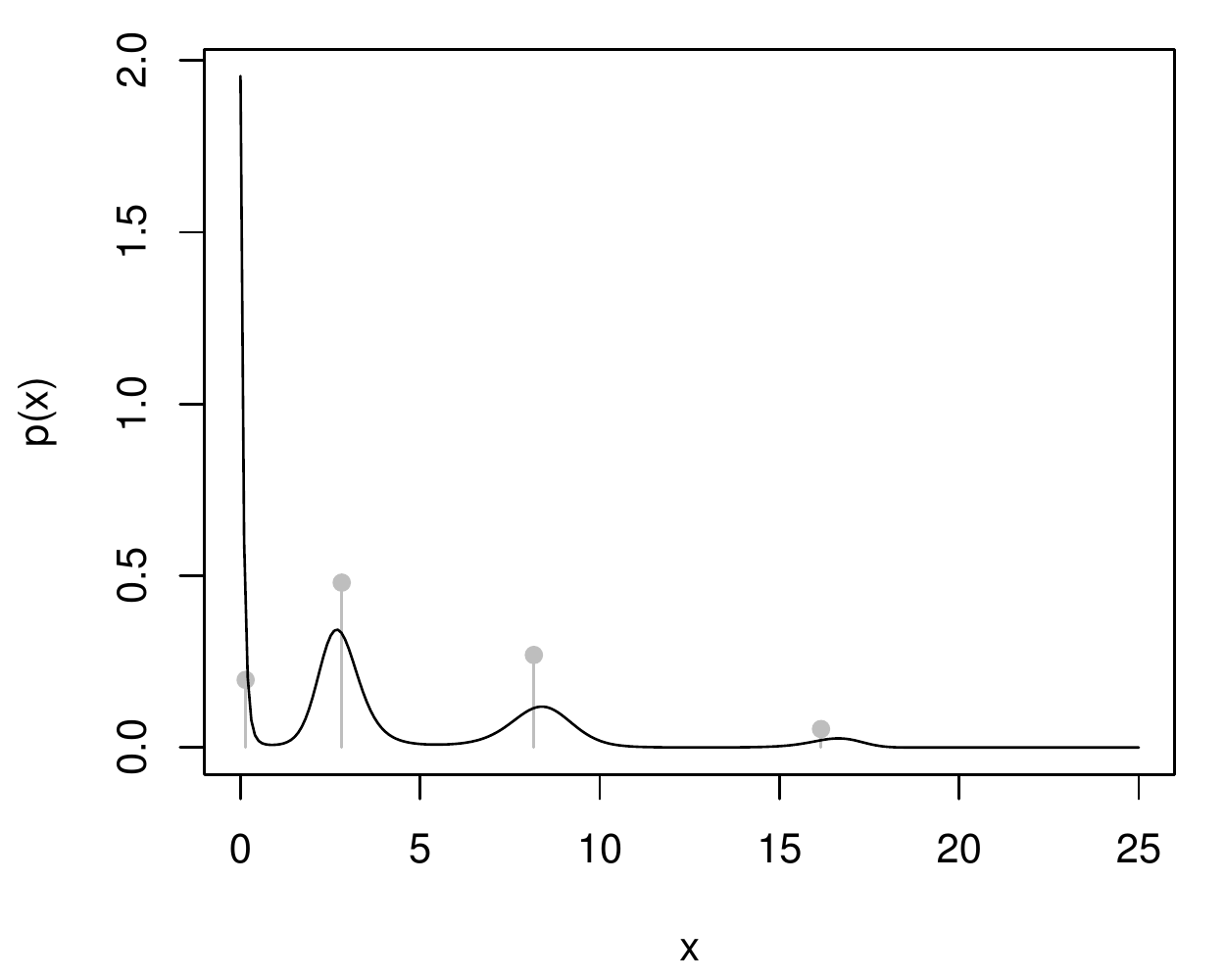}}}
\subfigure[$T=5000$ iterations]{\scalebox{0.6}{\includegraphics{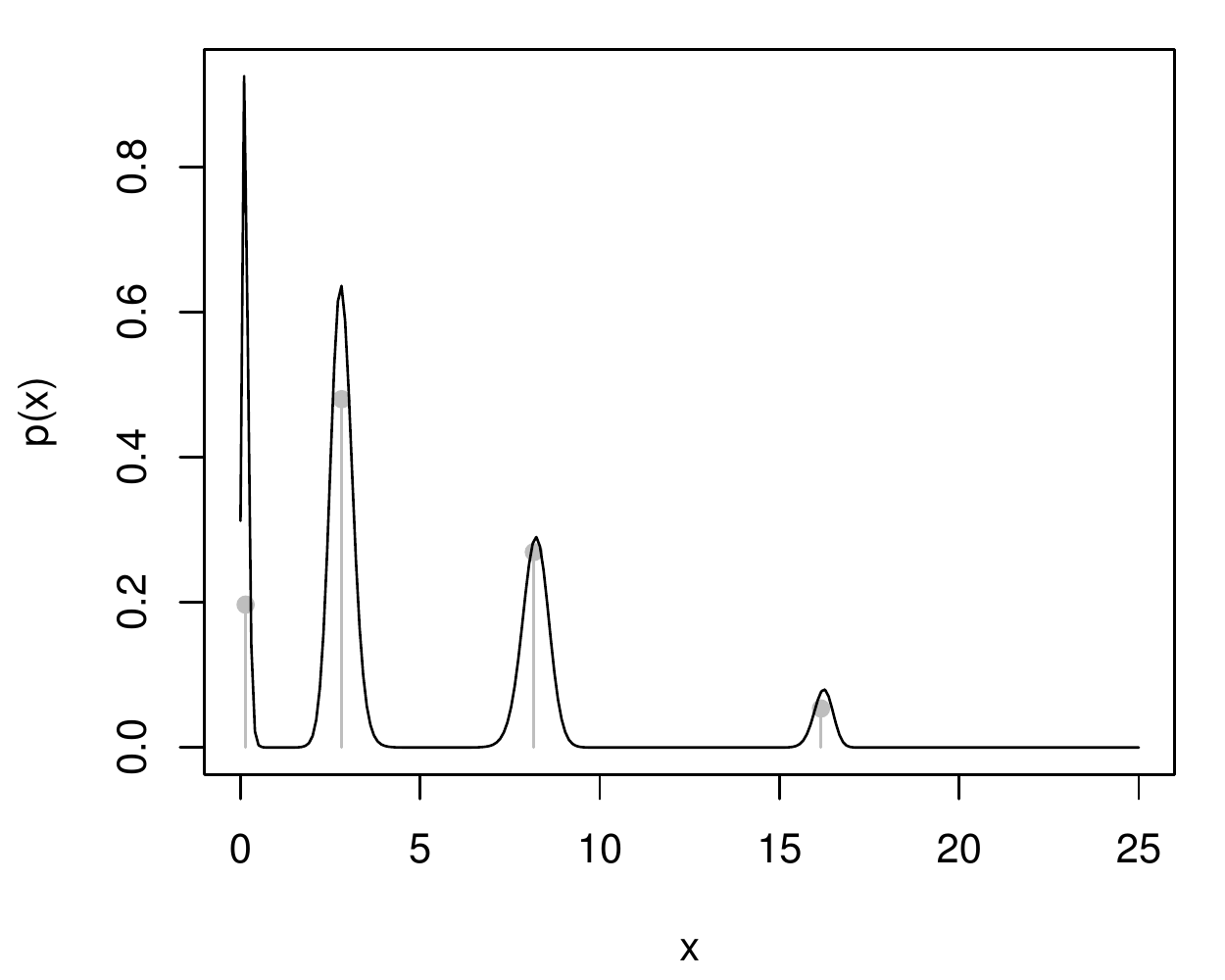}}}
\subfigure[Log-likelihood path, $t=1,\ldots,10$]{\scalebox{0.6}{\includegraphics{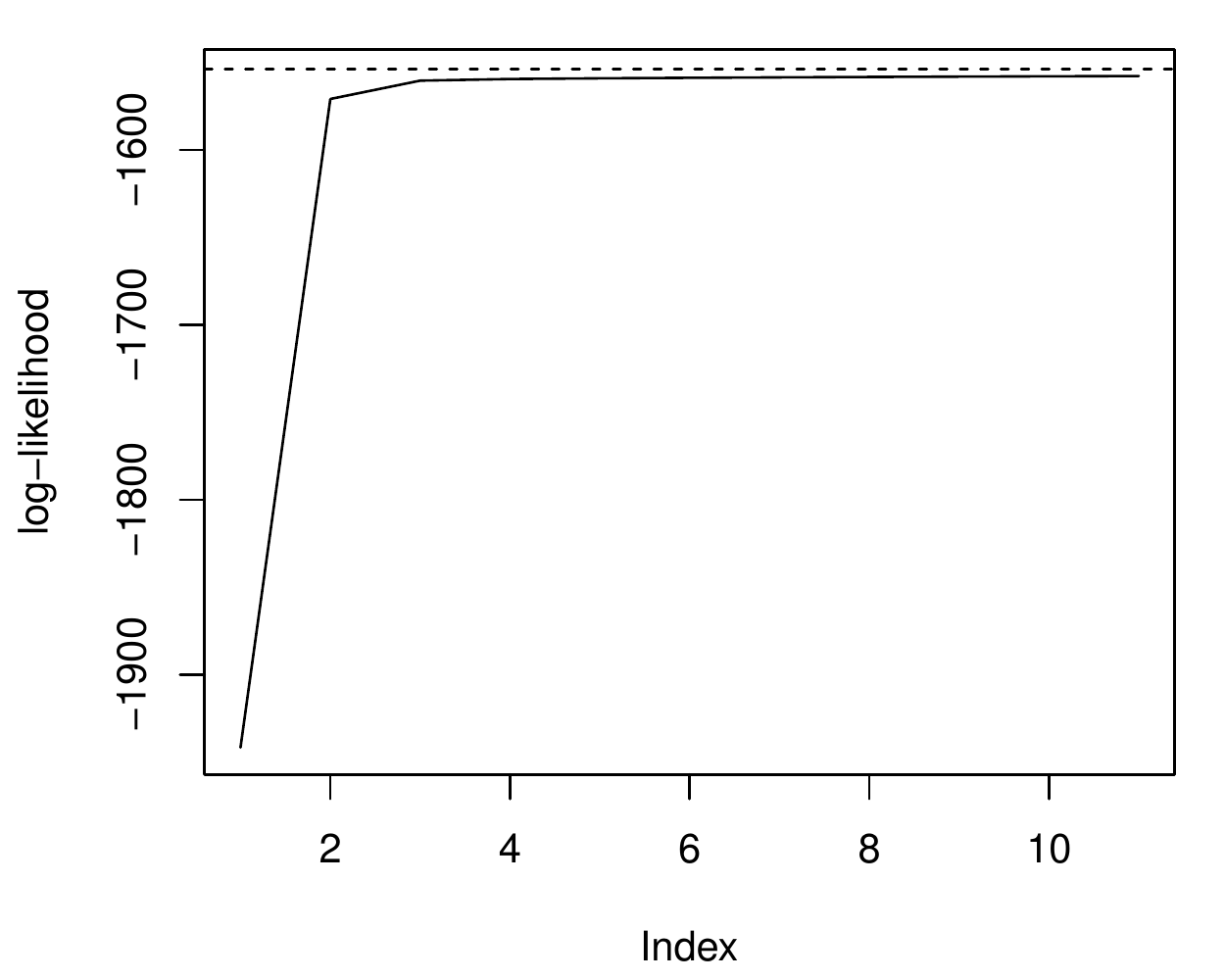}}}
\caption{Plot of the mixing density estimate $p_T$ for the Poisson mixture model for several values of $T$; vertical gray lines correspond to the nonparametric MLE.  Panel~(f) shows the path of the log-likelihood function compared to the maximum value (dashed line).}
\label{fig:thai}
\end{center}
\end{figure}

\subsection{Properties}
\label{SS:properties}

This section seeks to give a rigorous demonstration of two key results suggested in Sections~\ref{SS:definition}--\ref{SS:poisson} above.  First is the conjecture coming from Panel~(f) in Figure~\ref{fig:thai}, namely, that the likelihood function is non-decreasing as a function of $t$ in \eqref{eq:algorithm}, which implies that the iterations are not moving down on the likelihood surface.    

\begin{prop}
\label{prop:em}
For $(P_t)_{t=0}^\infty$ defined as \eqref{eq:algorithm}, the likelihood $t \mapsto L(P_t)$ is non-decreasing.
\end{prop}

\begin{proof}
Write the iterates in \eqref{eq:algorithm} as 
\[ p_{t+1}(x) = \int \frac{k(y \mid x) p_t(x)}{f_t(y)} \, \Fhat(dy), \]
where $\Fhat$ is the empirical distribution of $Y_1,\ldots,Y_n$.  Then the log-likelihood satisfies 
\begin{align*}
\ell(p_{t+1}) - \ell(p_t) & = \int \log \frac{f_{t+1}(y)}{f_t(y)} \, \Fhat(dy) \\
& = \int \left[ \log \int \frac{\int k(y \mid x) k(y' \mid x) \, p_t(x) \,dx}{f_t(y) \, f_t(y')} \, \Fhat(dy') \right]\, \Fhat(dy). 
\end{align*}
Applying Jensen's inequality, the right-hand side is lower-bounded by 
\[ \int \Bigl\{ p_t(x) \int \frac{k(y \mid x)}{f_t(y)} \, \Fhat(dy) \Bigr\} \log \Bigl\{ \int \frac{k(y' \mid x)}{f_t(y')} \, \Fhat(dy') \Bigr\} \, dx, \]
which is the Kullback--Leibler divergence of $p_t$ from $g_t(x) = p_t(x) \int \{k(y \mid x) / f_t(y) \} \, \Fhat(dy)$.  Since the latter is non-negative, it follows that $\ell(p_{t+1}) \geq \ell(p_t)$, as was to be shown. 
\end{proof}

We have opted here for a direct proof of non-decreasing likelihood, though this could also be verified by confirming that \eqref{eq:algorithm} is an EM algorithm with the parameter space $\PP$, the set of all probability measures $P$ on $\X$ having strictly positive densities.  While a characterization of \eqref{eq:algorithm} as an EM algorithm may not be new---e.g., \citet{laird.louis.1991} give a similar characterization of a special case of \eqref{eq:algorithm}---there appears to be no general results on whether the sequence actually converges to the nonparametric MLE, $\hat P$.  Of course, not moving downwards on the likelihood surface suggests that the algorithm converges to $\hat P$, but a proof requires some care.  Of critical importance is that, since $\PP$ is convex, the general results in \citet[][Sec.~5A]{csiszartusnady} imply that 
\begin{equation}
\label{eq:limit}
L(P_t) \to L(\hat P), \quad \text{as $t \to \infty$}, 
\end{equation}
a necessary condition for $P_t$ to converge to $\hat P$. 

\begin{prop}
\label{prop:limit}
Suppose for each $i=1,\ldots,n$ that $x \mapsto k(Y_i \mid x)$ is a continuous and strictly positive map on $\X$. Also, assume that for every $\epsilon > 0$ there exists a compact $\X_0 \subset \X$ such that $\sup_{x \in \X_0^c} k(Y_i \mid x) < \epsilon$ for each $i = 1, \ldots, n$. If there exists the unique nonparametric MLE $\hat P$, then $P_t \to \hat P$ weakly as $t \to \infty$.
\end{prop}

\begin{proof}
Since a continuous function on a compact set is bounded, the map $x \mapsto k(Y_i \mid x)$ is bounded on $\X$ for each $i=1, \ldots, n$. Since $\lim_{t\rightarrow\infty} L(P_t) > 0$ by \eqref{eq:limit}, we have $\liminf_t f_t(Y_i) > 0$ for each $i=1, \ldots, n$. For $\epsilon = \min_i k(Y_i \mid x)$, let $\X_0 \subset \X$ be a compact set such that $\sup_{x \in \X_0^c} k(Y_i \mid x) < \epsilon$ for every $i \leq n$. Then, there exists an integer $T$ such that $p_t(x) \leq p_T(x)$ for every $t \geq T$ and $x \in \X_0^c$. It follows that $(P_t)_{t \geq 0}$ is a uniformly tight sequence of probability measures.

By Prokhorov's theorem, every subsequence $P_{t_s}$ has a further subsequence $P_{t_{s(r)}}$ that converges weakly to a probability measure $P^\star$ in $\overline\PP$, the set of {\em all probability measures on $\X$}.  Since $x \mapsto k(Y_i \mid x)$ is continuous and bounded on $\X$, $P \mapsto L(P)$ is also continuous with respect to the topology of weak convergence on $\overline\PP$.  Therefore, 
\[ L(P_{t_{s(r)}}) \to L(P^\star), \quad r \to \infty. \]
It follows from \eqref{eq:limit} that $P^\star$ is a nonparametric MLE.  We assumed that $\hat P$ is the unique nonparametric MLE, so every subsequence of $P_t$ has a further subsequence converging weakly to $\hat P$, which implies $P_t \to \hat P$ weakly. 
\end{proof}

While convergence of \eqref{eq:algorithm} to the nonparametric MLE is important to explain and justify the method itself, the EM's convergence to the maximizer is known to be relatively slow \citep[e.g.,][]{jamshidian.jennrich.1997, meng.vandyk.1997}.  In fact, Figure~\ref{fig:thai} reveals that we need more than 5000 iterations of \eqref{eq:algorithm} in order to reach the MLE.  However, since we are interested in estimating a {\em density}, we prefer to stop the algorithm short of convergence to $\hat P$.  So Propositions~\ref{prop:em}--\ref{prop:limit}, along with the speed at which the iterations get close to the maximum, as demonstrated in Panel~(f) of Figure~\ref{fig:thai}, suggest that $p_T$, for suitable $T$, is a {\em near-maximum likelihood estimator} (nMLE) of the mixing density.

\section{A smooth nonparametric near-MLE}
\label{S:smoothed}

\subsection{Proposal}

As described above, for the purpose of maintaining smoothness of the mixing density estimator $\hat p_T$, we have reason to stop the algorithm at a relatively small value $T$.  But we want to ensure that the likelihood at $\hat p_T$ is sufficiently large and, hence, can be understood as a near-MLE, or nMLE.  Towards this, we propose the following stopping rule:
\begin{equation}
\label{eq:rule}
\text{stop at $T$ such that $\tilde \ell_{\text{ext}} - \ell(p_T) < \delta |\tilde \ell_{\text{ext}}|$},
\end{equation}
where $\delta$ is a small constant to be specified, and $\tilde \ell_{\text{ext}}$ is the log-likelihood evaluated at some external estimate of the mixture density $f$.  One option is for $\tilde \ell_{\text{ext}}$ to be the likelihood evaluated at the actual nonparametric MLE, but a reasonable alternative is to take $\tilde\ell_{\text{ext}}$ based on a simple kernel estimate of $f$.  In our experiments that follow, we take $\delta=0.05$ and $\tilde\ell_{\text{ext}}$ based on the default kernel estimator in the {\tt density} function in R.


\subsection{Numerical results}

Here we carry out a number of simulations to demonstrate the performance of our proposed nonparametric nMLE of the mixing density.  In each case, we take $Y_1, \ldots, Y_n$ to be an iid sample of size $n=500$ from the mixture density $f(y) = \int k(y \mid x) p(x) \, dx$.  Three different kernels will be considered:
\begin{description}
\item[\sc Kernel 1.] $k(y \mid x) = {\sf N}(y \mid x, \frac12)$;
\item[\sc Kernel 2.] $k(y \mid x) = \frac{1}{0.3} \, {\sf t}(\frac{y-x}{0.3} \mid \text{df} = 5)$;
\item[\sc Kernel 3.] $k(y \mid x) = {\sf Gamma}(y \mid \text{shape}=20x, \, \text{rate}=20)$.  
\end{description}
We will also consider three different mixing densities, supported (roughly) on $[0,10]$:
\begin{description}
\item[\sc Mixing 1.] $p(x) = \frac{1}{10} \, {\sf Beta}(\frac{x}{10} \mid \text{shape}=5, \, \text{shape}=5)$;
\item[\sc Mixing 2.] $p(x) = \frac34 \, {\sf N}(x \mid 3, 0.8^2) + \frac14 \, {\sf N}(x \mid 7, 0.8^2)$;
\item[\sc Mixing 3.] $p(x) = {\sf Gamma}(x \mid \text{shape}=2, \, \text{rate}=1)$.  
\end{description}

For each kernel and mixing density pair, 100 datasets are generated, and mixing density estimates using the proposed nMLE based on stopping rule \eqref{eq:rule} above; all are based on a ${\sf Unif}(0,10)$ starting initial guess.  Figure~\ref{fig:sim} summarizes the results of these simulations, plotting the 100 mixing density estimates for the individual datasets, the truth, the point-wise average of those 100 estimate, and the point-wise one-standard deviation range around the average.  In all cases, although there is variability around the truth in each individual case, as is expected, our proposed estimator is quite accurate overall.  Also, the $T$ defined by \eqref{eq:rule} was less than 5 in all of these runs.  For comparison, we also produced the mixing density estimator based on the predictive recursion algorithm described briefly in Section~\ref{S:intro}.  Figure~\ref{fig:tvbox} shows the predictive recursion to nMLE relative $L_1$ error.  Since the values all tend to be greater than one for all kernel and mixing density pairs, we conclude that nMLE tends to be more accurate than predictive recursion.

\begin{figure}
\begin{center}
\subfigure[Mixing 1, Kernel 1]{\scalebox{0.4}{\includegraphics{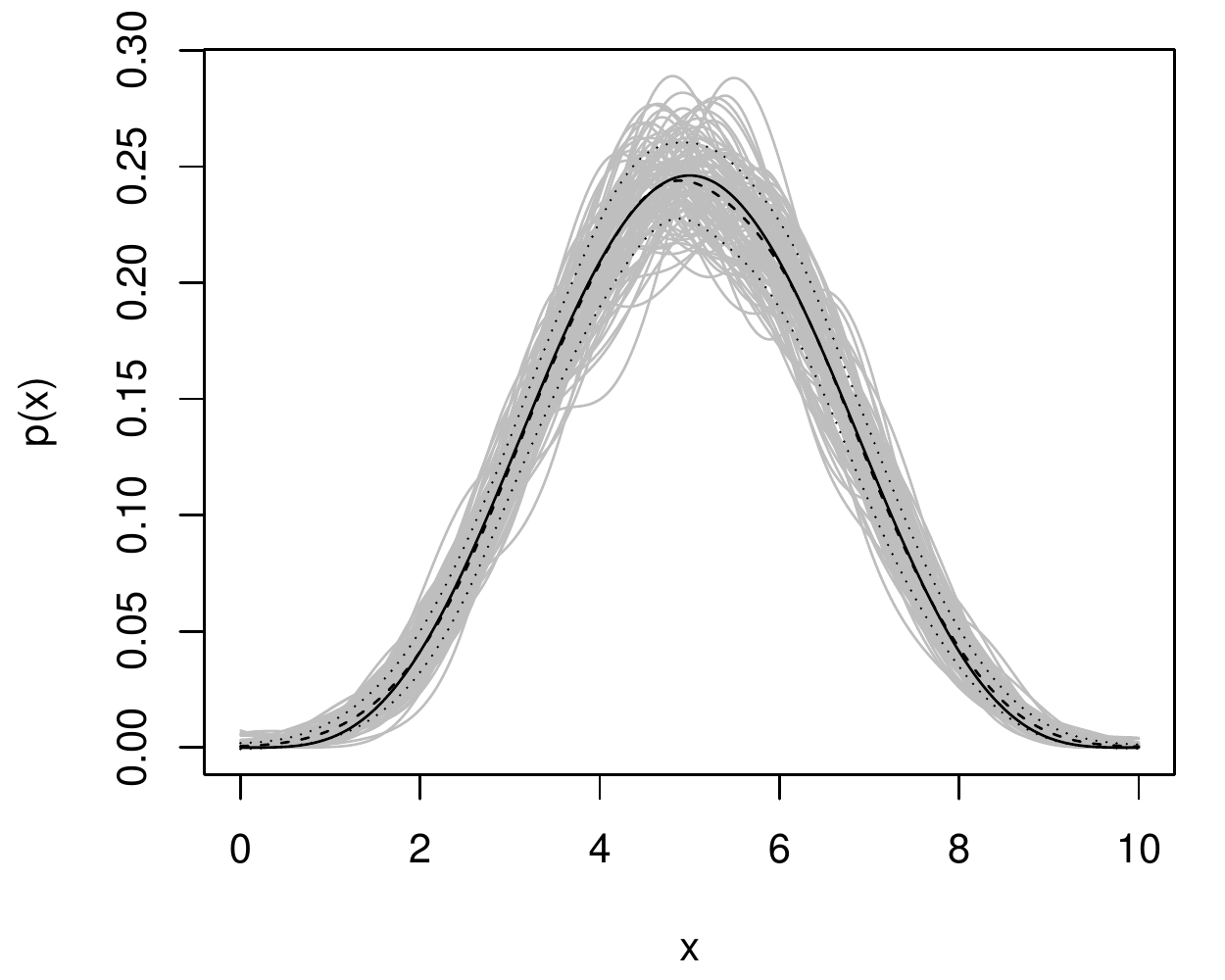}}}
\subfigure[Mixing 1, Kernel 2]{\scalebox{0.4}{\includegraphics{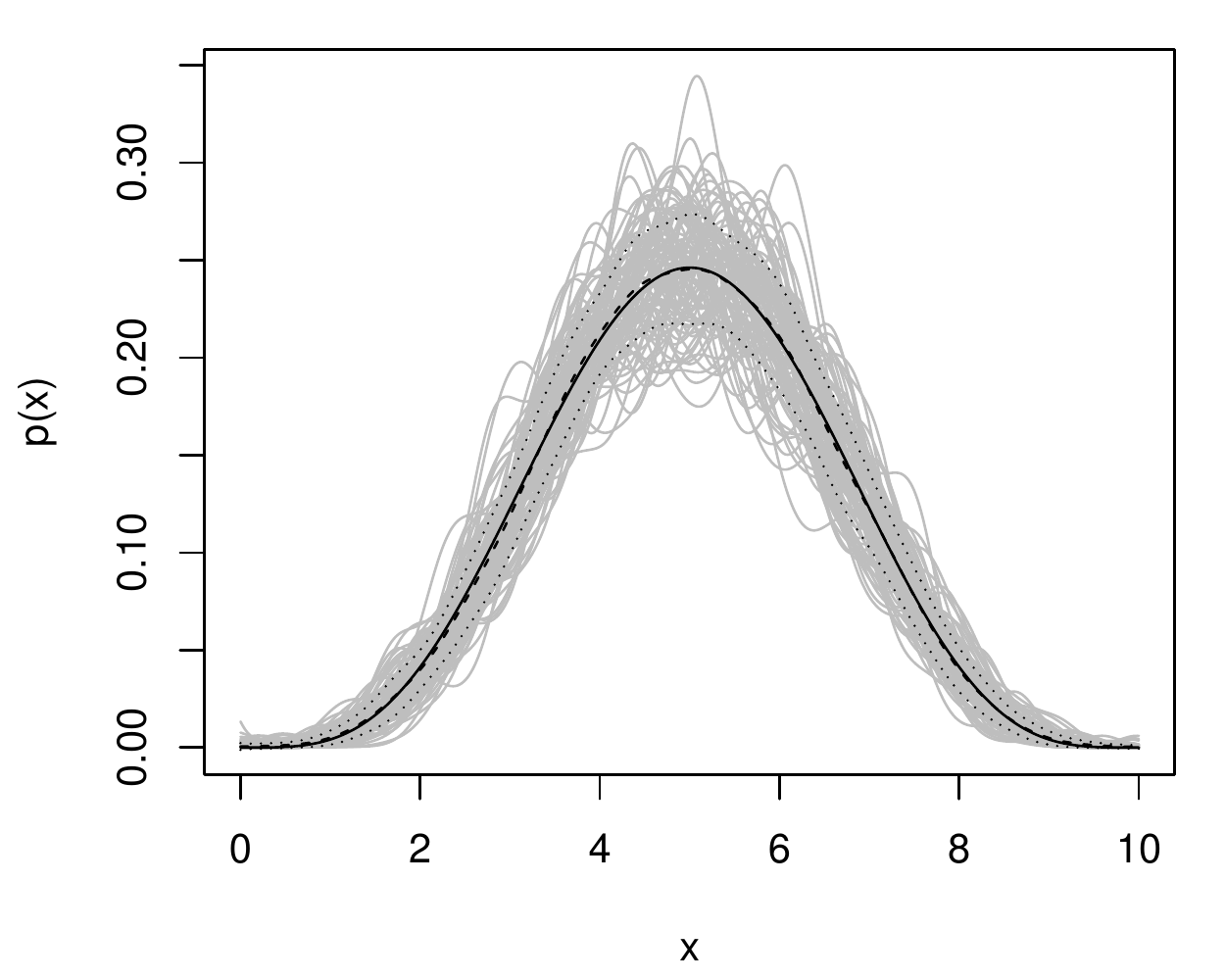}}}
\subfigure[Mixing 1, Kernel 3]{\scalebox{0.4}{\includegraphics{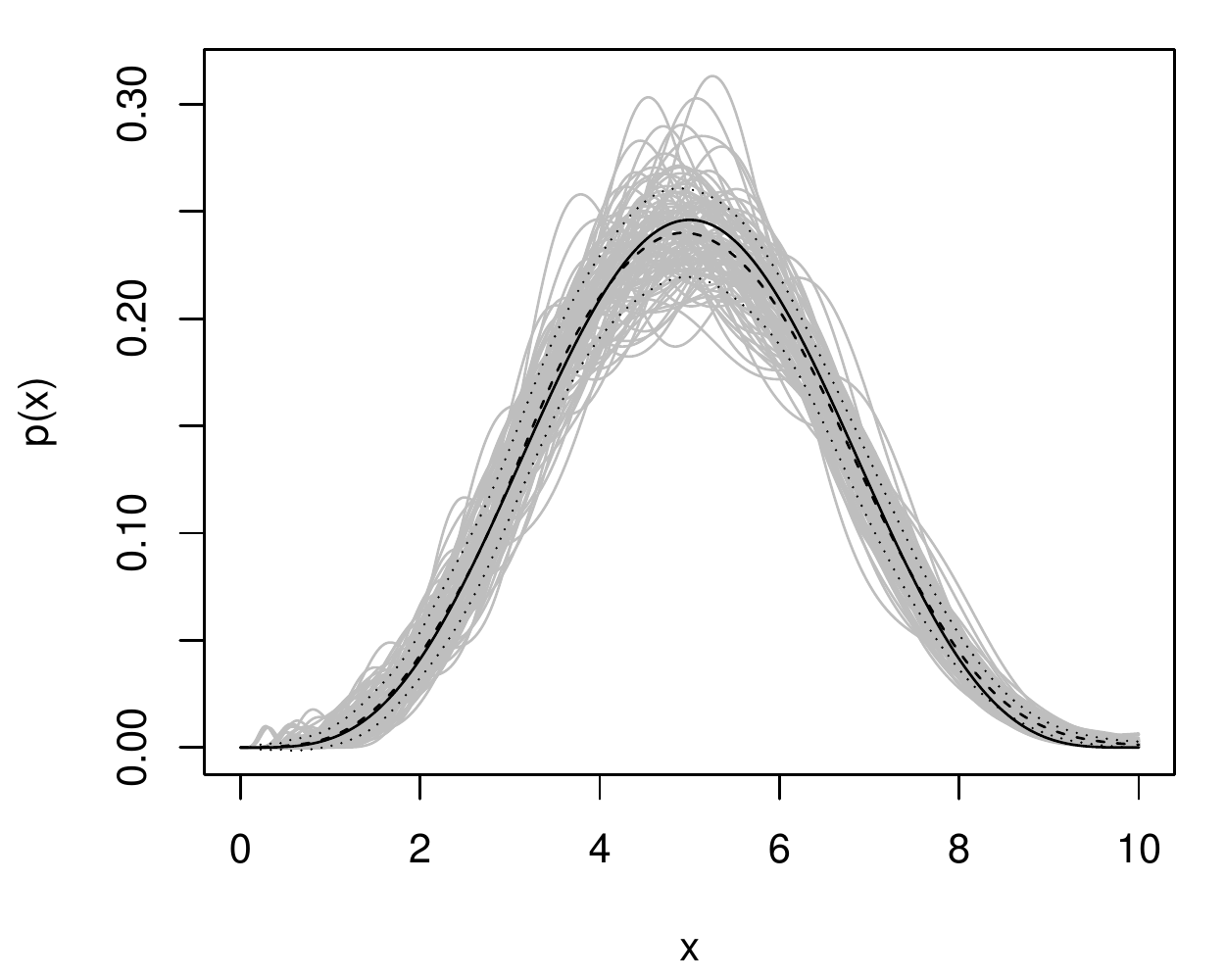}}}
\subfigure[Mixing 2, Kernel 1]{\scalebox{0.4}{\includegraphics{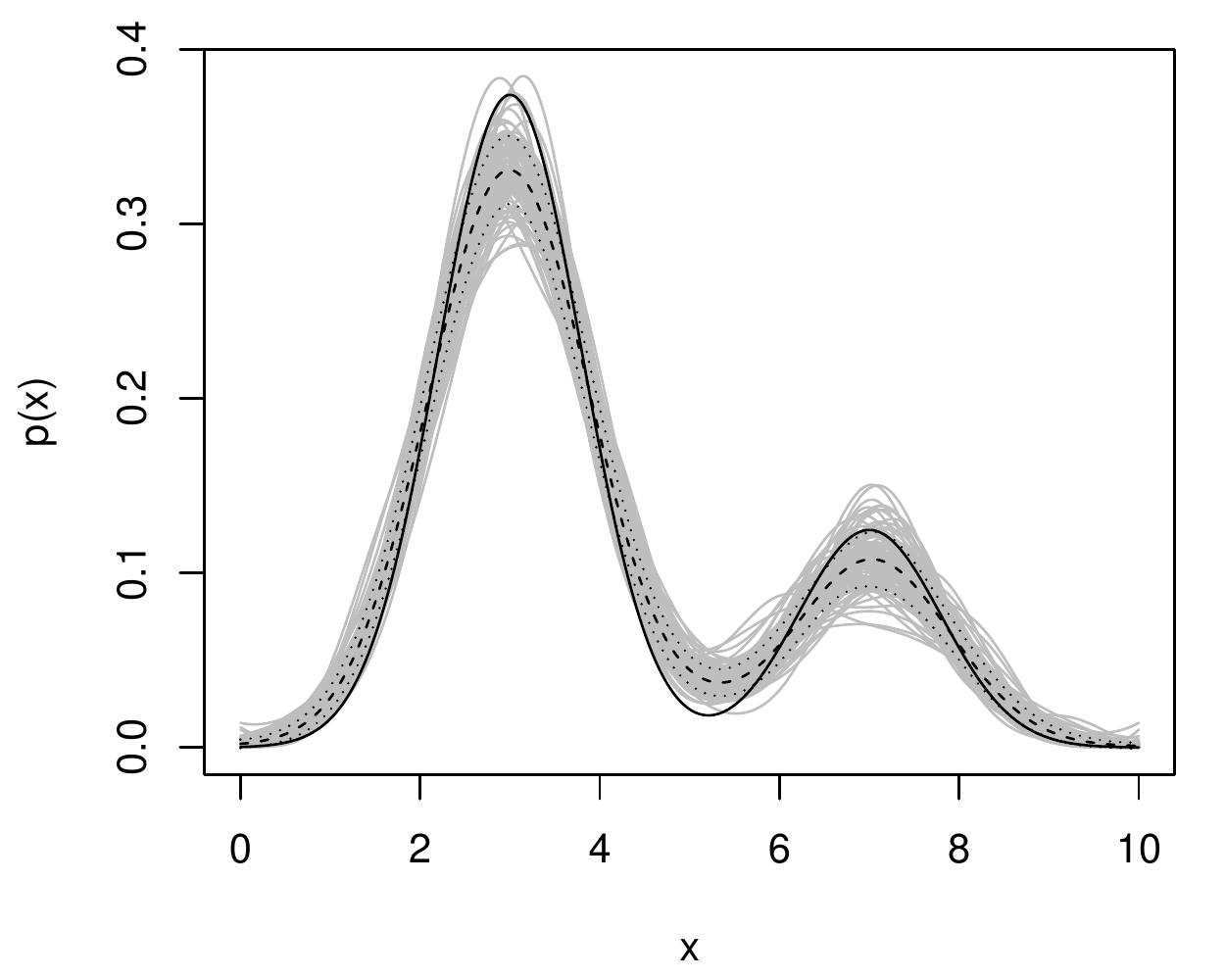}}}
\subfigure[Mixing 2, Kernel 2]{\scalebox{0.4}{\includegraphics{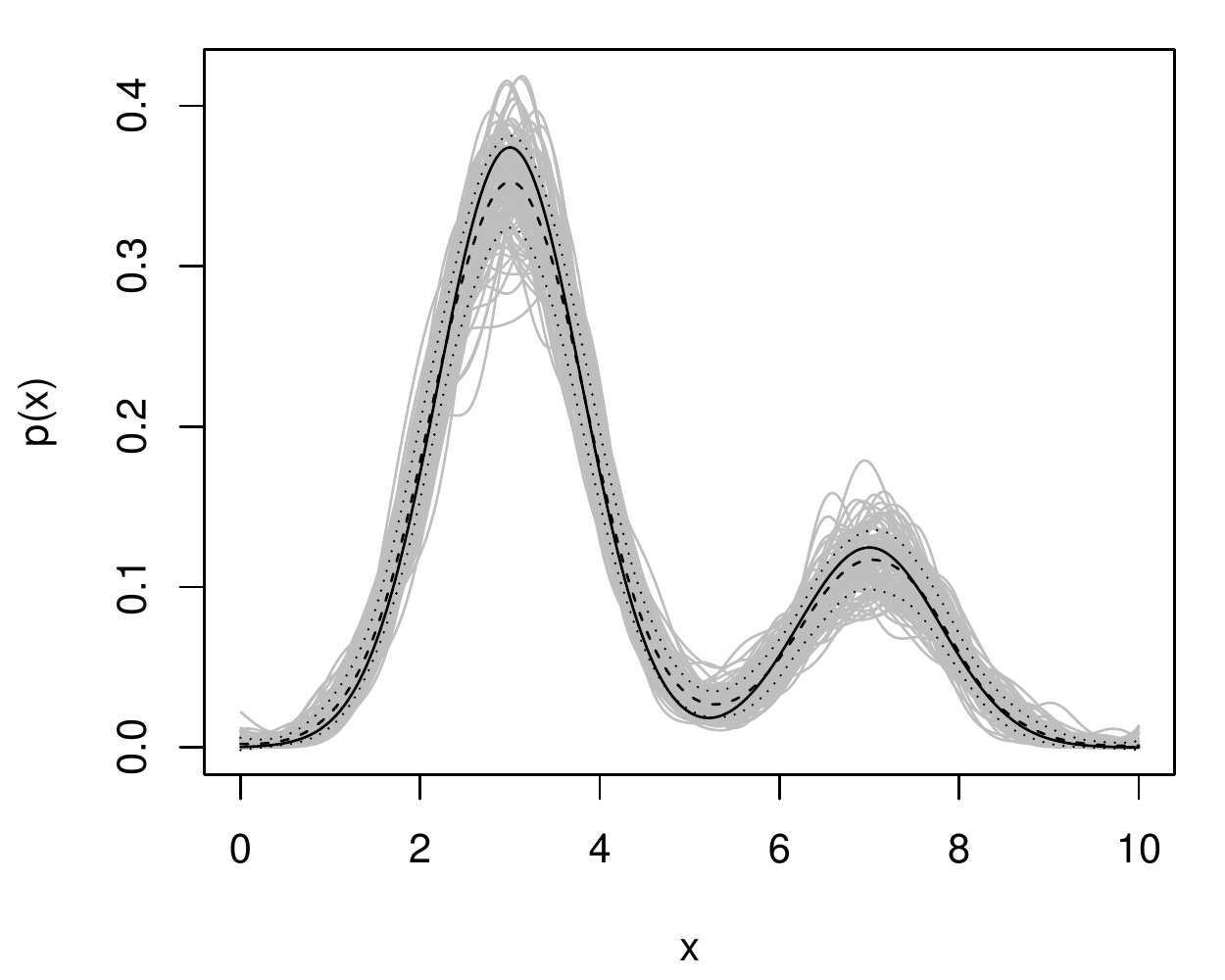}}}
\subfigure[Mixing 2, Kernel 3]{\scalebox{0.4}{\includegraphics{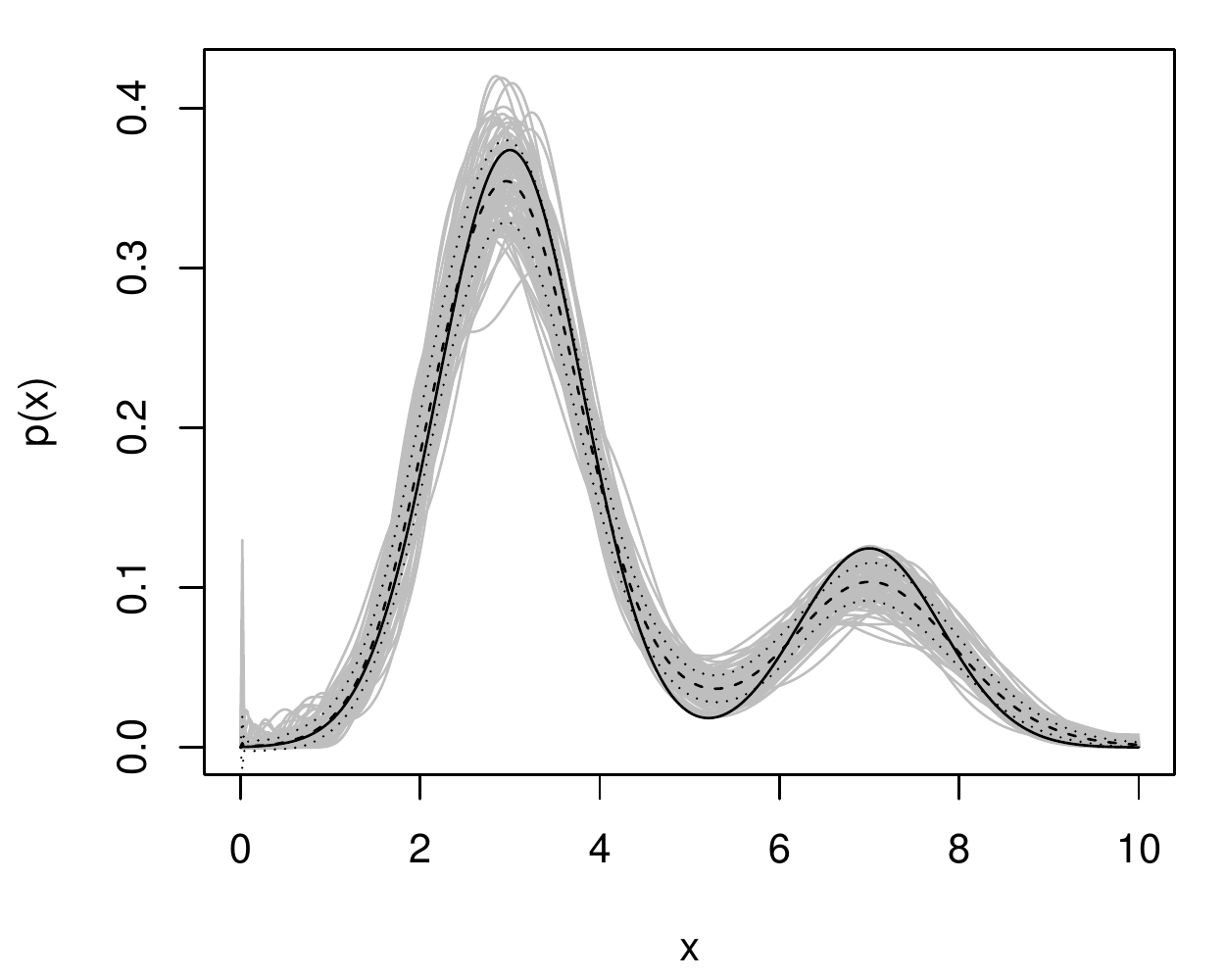}}}
\subfigure[Mixing 3, Kernel 1]{\scalebox{0.4}{\includegraphics{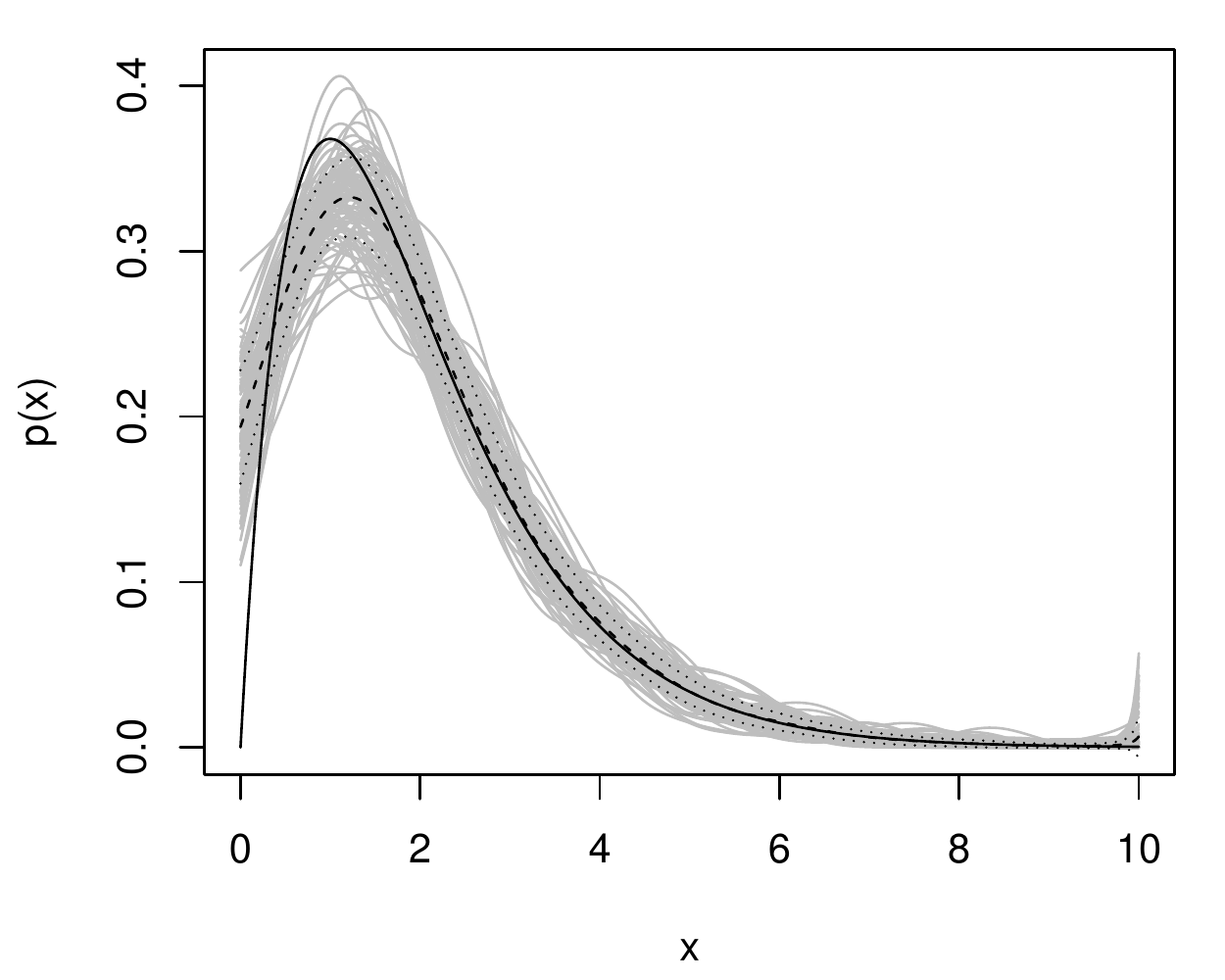}}}
\subfigure[Mixing 3, Kernel 2]{\scalebox{0.4}{\includegraphics{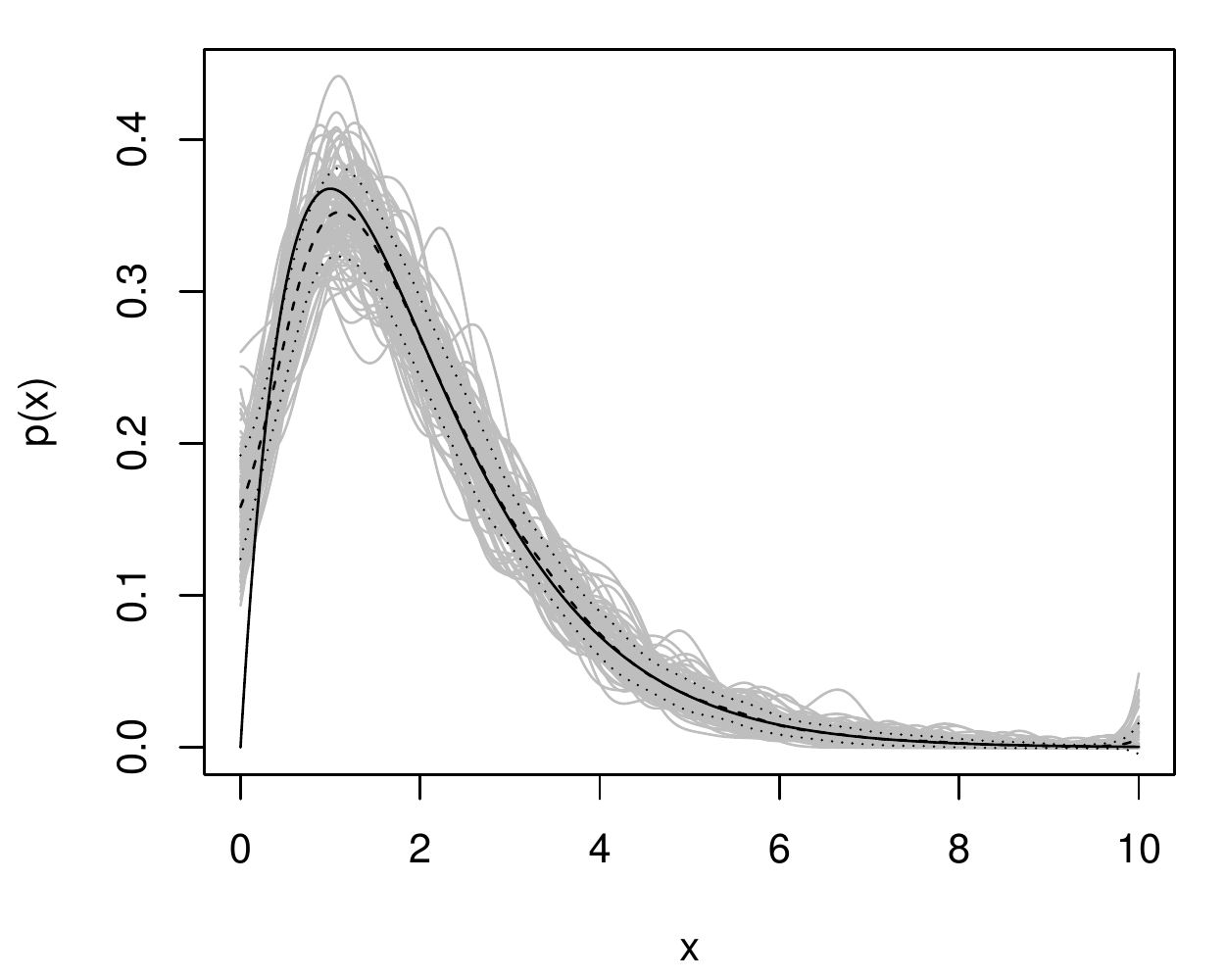}}}
\subfigure[Mixing 3, Kernel 3]{\scalebox{0.4}{\includegraphics{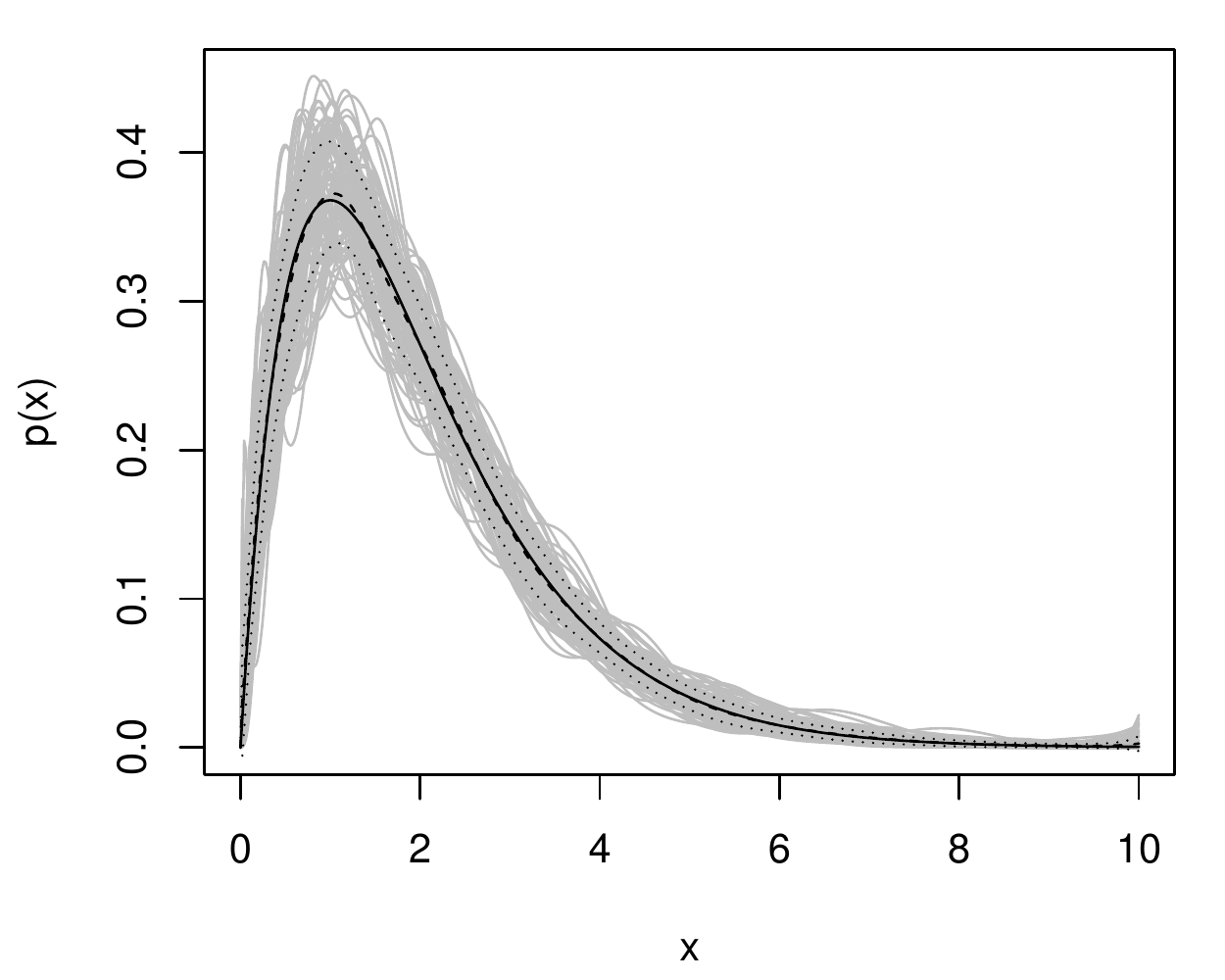}}}
\caption{For each mixing and kernel density pair, gray lines are the individual estimates, solid line is the true mixing density, dashed line is the point-wise average, and dotted lines are the point-wise one-standard deviation range around the average.}
\label{fig:sim}
\end{center}
\end{figure}

\begin{figure}
\begin{center}
\scalebox{0.75}{\includegraphics{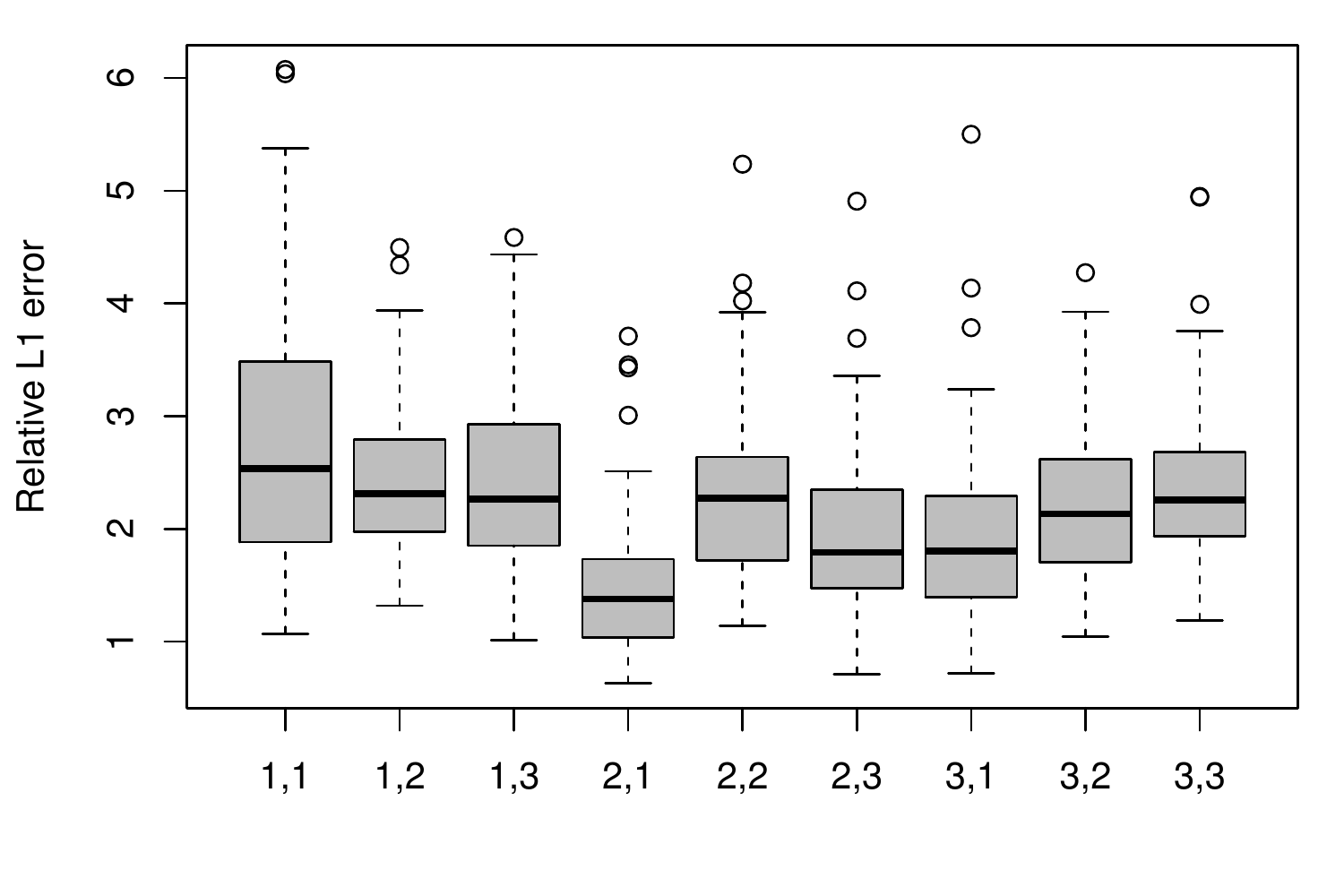}}
\caption{Side-by-side boxplot of the relative $L_1$ errors comparing the nMLE to the predictive recursion estimator, for each kernel and mixing density pair.}
\label{fig:tvbox}
\end{center}
\end{figure}


\section*{Acknowledgments}

This work is partially supported by the National Science Foundation under grants DMS--1612891 and DMS--1737929

\ifthenelse{1=0}{
\bibliographystyle{apalike}
\bibliography{/Users/rgmarti3/Dropbox/Research/mybib}
}{

}

\end{document}